\documentclass{article}
\pdfoutput=1
\usepackage{amsmath,amsfonts}
\usepackage{amssymb}
\usepackage{amsthm}

\newtheorem{theorem}{\bf{Theorem}}[section]
\newtheorem{problem}{\bf{Problem}}[section]

\newcommand{\footremember}[2]{%
    \footnote{#2}
    \newcounter{#1}
    \setcounter{#1}{\value{footnote}}%
}
\newcommand{\footrecall}[1]{%
    \footnotemark[\value{#1}]%
}

\title{A Simpler NP-Hardness Proof for Familial Graph Compression}

\author{Ammar Ahmed\footremember{itu}{Department of Computer Science, Information Technology University,  Pakistan} \footremember{contrib}{All authors contributed equally to this article.} \and Zohair Raza Hassan\footrecall{itu} \footrecall{contrib} \and Mudassir Shabbir\footrecall{itu} \footrecall{contrib}}
\date{{\footnote{\textit{Email addresses:} ammar.ahmed@itu.edu.pk (Ammar Ahmed), zohair.raza@itu.edu.pk (Zohair Raza Hassan), mudassir@rutgers.edu (Mudassir Shabbir).}}}

\begin{document}
\maketitle

\begin{abstract}
    This document presents a simpler proof showcasing the NP-hardness of Familial Graph Compression.
\end{abstract}

\section{Introduction}

Familial Graph Compression (FGC) is a problem introduced in~\cite{ahmed2020interpretable}. The problem entails determining whether it is possible to convert a given graph $G$ to a target graph $H$ via a series of ``compressions'' based on the presence of certain sub-graphs in $G$, specified in a set $\mathcal{F}$. A complete definition is given in the next section. A single instance of FGC involves $G$, $H$, and $\mathcal{F}$ as input. This problem was proven to be \texttt{NP-complete} in~\cite{ahmed2020interpretable}:

\begin{theorem}
    \label{thm:FGC_NPC}
    The FGC problem is \texttt{NP-complete} when:
    \begin{enumerate}
        \item $G$ is simple graph on $n$ nodes, $H$ is the single node graph, and family $\mathcal{F}$ contains a single motif $C_n$ i.e. a cycle on $n$ nodes.
        \item $G$  is a simple graph on $n=3k$ nodes, $H$ is the single node graph, and $\mathcal{F}$ contains a single motif with $k$ disjoint triangles.
        \item $G$ is a simple graph, $H$ is a forest of isolated nodes, and $\mathcal{F}$ is a family of graphlets.
    \end{enumerate}
\end{theorem}

In this work, we provide an easier proof for the third setting. 

\section{Notation and Terminology}

We adopt the same notation and terminology as in~\cite{ahmed2020interpretable}. The relevant preliminaries have been reiterated below.

\subsection{Preliminaries}

A graph $G$ is a collection of nodes $V$ and edges $E\subseteq V\times V$ i.e. pairwise interactions between pairs of nodes. 
For a node $u$, its neighborhood $N(u)$ is defined as the set of all nodes $v\in V$ such that there exists an edge $(u,v)$ in $E$. The degree $d(u)$ is defined as the size of the neighborhood of a node $u$.
$G$ is undirected and unweighted, i.e. for $u,v\in V$, an edge $(u,v)$ is same as the edge $(v,u)$. 
For a fixed graph $G=(V,E)$, a given $F = (V_F, E_F)$ is called a \textit{motif} of $G$,
if $F$ is isomorphic to a sub-graph in $G$ i.e. $F$ is a motif if there exists $V'\subset V$ and a function $\phi : V_F \rightarrow V'$ such that for all edges $(u,v)\in E_F$ there is an edge $(\phi(u), \phi(v))\in E$. Similarly, $F = (V_F, E_F)$ is called a \textit{graphlet} of $G$,
if $F$ is isomorphic to an \textit{induced} sub-graph in $G$ i.e. $F$ is a graphlet if there exists $V'\subset V$ and a function $\phi : V_F \rightarrow V'$ such that for all edges $(u,v)\in E_F$ \textit{if and only if} there is an edge $(\phi(u), \phi(v))\in E$. We will use the term motif (and similarly graphlet) for both $F$ and any of its isomorphic copies in $G$.

For a given equivalence relation $\sim$ on the set nodes of a graph $G$, the quotient graph, denoted by $G\big/\sim$, is a graph where the node set is the set of equivalence classes defined by $\sim$ and there is an edge between a pair of nodes (classes) if and only if there is an edge between any pair of nodes of two corresponding classes in $G$. Intuitively, in quotient graphs, prescribed subsets of nodes are merged and the incidence is preserved without creating multi-edges~\cite{golumbic2006graph}.
We will repeatedly deal with graphs with names $G,H$, and $F_i$; their node and edge set will, respectively, be denoted by $(V_G,E_G)$, $(V_H,E_H)$ and $(V_{F_i},E_{F_i})$. Finally, for a set $V$ and a positive integer $c$, $\binom{V}{c}$ is defined as the set of all size subsets of $V$ with exactly $c$ elements.

\subsection{Familial Graph Compression}

We start by defining an equivalence relation on the node set $V$ of $G$ based on a motif (or a graphlet) $F$. Consider the relation $R_F$ where node $u$ is related to $v$ whenever both $u$ and $v$ lie in a sub-graph of $G$ isomorphic to $F$. We define $\sim_F$ to be the transitive closure of $R_F$. 
Intuitively, if two motifs (resp. graphlets) share a common node in $G$, then all nodes in both motifs (resp. graphlets) are related in $\sim_F$. Clearly, $\sim_F$ is an equivalence relation on $V$. Then, an $F$-\textit{compression step} (referred to as compression step when $F$ is clear from the context) is defined as computing the quotient graph $G\big/\sim_F$. Recall that a quotient graph $G\big/\sim_F$ is a graph on classes in the partition $\sim_F$, where two classes are adjacent if any pair of nodes in the corresponding classes are adjacent in the graph $G$.
The \textit{familial compression} of a graph $G$ for a family $\mathcal{F}$ is the process of repeatedly applying $F_i$-compression steps on $G$ where after each step $G$ is replaced by the quotient graph of the previous step. Thus, we say that a graph $H$ can be constructed by a $\mathcal{F}$-compression of $G$ if there exist a sequence of graphs:
$[G^0\;\;G^1\;\;G^2\;\ldots \;G^k=H]$ where $G^0=G$ and $G^j=G^{j-1}\big/\sim_{F_{i}}$ i.e. $G^j$ is result of an $F_i$-compression on the graph $G^{j-1}$ for some $F_i\in \mathcal{F}$. Note, that a graph $H$ may be constructed in several different ways via different compression steps. To avoid trivial compressions, we restrict that each $F\in \mathcal{F}$ contains at least three nodes. 
The following is the FGC problem:
\begin{problem}[Familial Graph Compression]
Given simple graphs, $G$, and $H$, and a family of motifs (or graphlets) $\mathcal{F}$, can $H$ be constructed from a $\mathcal{F}$-compression of $G$?
\end{problem}

\section{Result}

In the original proof for Theorem~\ref{thm:FGC_NPC}-(3), a reduction is provided from a variant of the 3-SAT problem to FGC. In this section we showcase the same result via reduction from Exacty Cover by Three Sets (XC3), defined below.

\begin{problem}[Exact Cover by Three Sets~\cite{gonzalez1985clustering}]
Let $X = \{x_1, x_2, \ldots, x_{3k}\}$, and let $S$ be a collection of 3-element subsets of $X$, in which no element in $X$ appears in more than three subsets. For $s_j \in S, s_j = \{ x_{j_1}, x_{j_2}, x_{j_3} \}$. The problem consists of determining whether $S$ has an exact cover for $X$, i.e. a $S' \subseteq S$ such that every element in $X$ occurs in exactly one member of $S'$.
\end{problem}

This problem was proven to be \texttt{NP-complete} in~\cite{gonzalez1985clustering}. Note that for our reduction, the fact that ``each element appears in no more than three subsets'' is inconsequential. 

\begin{theorem}

\label{thm:fgc-xc3}
XC3 $\leq_{P}$ FGC.
\end{theorem}
\begin{proof}
Suppose we are given an instance of XC3, i.e. the sets $X$ and $S$. We show how one can make graphs $G$, and $H$, and family $\mathcal{F}$ for an FGC instance that is solvable only if the given XC3 instance is solvable.

Let $C_i$ denote a cycle on $i$ vertices. Let $f(i) = i + 2$ for $i \in \{ 1,2,3,\ldots \}$. The graph $G$ is the union of $3k$ disjoint cycles: $G = \bigcup_{x_i \in X} C_{f(i)}$.
For each $s_j \in S$, we define a graph $Z_j$ which is the union of three disjoint cycles: $Z_j = C_{f(j_1)} \cup C_{f(j_2)} \cup C_{f(j_3)}$. The family $\mathcal{F}$ contains $Z_j$ for each $s_j \in S$: $\mathcal{F} = \bigcup_{s_j \in S} Z_j$.
Finally, the target graph $H$ is a graph on $k$ isolated vertices, i.e. $|V_H| = k$, and $E_H = \emptyset$.

Intuitively, when a $Z_j$ is compressed in $G$, it corresponds to selecting a $c_j \in S$ to form an exact cover for $X$.
Observe that FGC would not allow the same element to be covered by different $c_j$'s, since the cycle corresponding to the covered elements no longer exist in the quotient graph, and thereby can't be compressed (selected) again. We get $k$ isolated vertices if an only if $k$ disjoint 3-element subsets form an exact cover of $X$. Clearly, the reduction can be performed in polynomial time.
\end{proof}

Observe that the $G$, $H$, and $\mathcal{F}$ used in Theorem~\ref{thm:fgc-xc3} are exactly as described in Theorem~\ref{thm:FGC_NPC}-(3). We note that this reduction holds even when $\mathcal{F}$ is a family of motifs. We also obvserve that some simple changes to the provided reduction can be made to show the following:
\begin{theorem}
\label{thm:fgc-4}
FGC is \texttt{NP-complete} when $G$ is a connected, simple graph, H is the single node graph, and $\mathcal{F}$ is a family of graphlets or motifs.
\end{theorem}

\bibliography{main}
\bibliographystyle{ieeetr}

\end{document}